%% file: MLU.tex
\documentclass[11pt]{article}

\usepackage[margin=1in]{geometry}
\usepackage{epsfig} 
\usepackage{amssymb}
\usepackage{amsmath}
\let\labelindent\relax
\usepackage{enumitem}
\usepackage{pifont}
\usepackage[dvipsnames]{xcolor}
\usepackage{tikz}
\usepackage{pgfplots}
\usepackage{algorithm2e}
\usepackage{varwidth}
\usepackage{soul}
\usepackage{verbatim}
\usetikzlibrary{arrows,shapes}
\usepackage{comment}

\usepackage{graphics} 

\usepackage{amsthm}
\usepackage[scientific-notation=true]{siunitx}
\usepackage[colorlinks]{hyperref}
\usepackage{amsfonts}
\usetikzlibrary{arrows.meta}
\usepackage{dsfont}

\usepackage{graphicx}
\usepackage{caption}
\usepackage{subcaption}
\usepackage{multirow}
\usepackage{setspace}
\usepackage{cite}
\usepackage{lineno}

\newtheorem{theorem}{Theorem}

\newtheorem{lemma}{Lemma}

\newtheorem{proposition}{Proposition}
\theoremstyle{definition}

\newtheorem{definition}{Definition}

\theoremstyle{remark}

\usepackage{placeins}

\newcounter{tmp}

\newcommand{\bb}[1]{\mathbf{#1}}

\definecolor{darkblue}{RGB}{0,101,204}
\definecolor{carorange}{RGB}{255,131,0}

\title{\LARGE \bf
 An Extended Game--Theoretic Model for Aggregate Lane Choice Behavior of Vehicles at Traffic Diverges with a Bifurcating Lane}

\author{Ruolin Li, Negar Mehr and Roberto Horowitz
\thanks{{Ruolin Li, Negar Mehr and Roberto Horowitz are with the Mechanical Engineering Department, University of California, Berkeley, CA, USA.
	{\tt\small \{ruolin\_li, negar.mehr, horowitz\}@berkeley.edu}}
}
}

\begin{document}
\maketitle

\thispagestyle{empty}
\pagestyle{empty}

\begin{abstract}
Road network junctions, such as merges and diverges, often act as bottlenecks that initiate and exacerbate congestion.
More complex junction configurations lead to more complex driver behaviors, resulting in aggregate congestion patterns that are more difficult to predict and mitigate.
In this paper, we discuss diverge configurations where vehicles on some lanes can enter only one of the downstream roads, but vehicles on other lanes can enter one of several downstream roads.
Counterintuitively, these bifurcating lanes, rather than relieving congestion (by acting as a versatile resource that can serve either downstream road as the demand changes), often cause enormous congestion due to lane changing.
We develop an aggregate lane--changing model for this situation that is expressive enough to model drivers' choices and the resultant congestion, but simple enough to easily analyze.
We use a game--theoretic framework to model the aggregate lane choice behavior of selfish vehicles as a Wardrop equilibrium (an aggregate type of Nash equilibrium).
We then establish the existence and uniqueness of this equilibrium.
We explain how our model can be easily calibrated using simulation data or real data, and we present results showing that our model successfully predicts the aggregate behavior that emerges from widely--used behavioral lane--changing models.
Our model's expressiveness, ease of calibration, and accuracy may make it a useful tool for mitigating congestion at these complex diverges.

\end{abstract}

\input{introduction.tex}
\input{model.tex}
\input{eq_properties.tex}
\input{model_validation.tex}
\input{conclusion.tex}

\section*{Acknowledgments}
This work was supported by the National Science Foundation under Grant
CPS 1545116. The authors thank Matthew Wright for his helpful input during the preparation of this article.


\bibliographystyle{IEEEtran}
\bibliography{IEEEabrv,MLU}

\end{document}

%% file: introduction.tex
\section{Introduction}\label{intro}
Road traffic congestion is a major source of inefficiency in modern society.
One study~\cite{urban_mobility_2015} estimated that, in 2014 in the U.S. alone, delays due to congestion cost drivers over 7 billion hours, had a social cost of US~\$160 billion, and led to the burning of 3 billion extra gallons of fuel.

In the transportation community, it is well known that vehicles' lane change maneuvers can be a significant cause of congestion. Unfortunately, lane changes are notoriously difficult to model, both because it is difficult to predict when and why a driver might change lanes \emph{and} how their maneuver will affect the movements of other vehicles~\cite{zheng_recent_2014}.  In what follows, we briefly
review relevant papers to our work, beginning with papers that analyze lane changing from the microscopic perspective and following with papers that analyze it from the macroscopic perspective.

A significant number of papers that analyze lane changes at the  ``microscopic level'' focus on determining accurate yet simple  driver behavior models that will be able to reproduce actual individual vehicle's lane changing behavior on a variety of scenarios. In~\cite{gipps1986model}, a model of a driver's lane changing decision-making process is formulated, focusing on  decisions that balance safety concerns with lane-changing incentives. These decisions require a series of evaluations concerning the velocity and proximity of the vehicles that surround the {\it ego vehicle} in the lane that it wishes to change to. Much work, both theoretical and simulation-based, has taken place on so-called ``gap acceptance models'', which model whether a driver will attempt to change lanes as a function of the inter-vehicle gaps that arise in their target lane (see~\cite{halati1997corsim,kazi1999acc,toledo2003lanechanging}, among others). Subsequently, researchers have explored a variety of rule-based microscopic models. In~\cite{kesting2007general}, for example, a new lane changing model, MOBIL, is proposed to minimize the overall braking induced by lane change. Microscopic level research in vehicles' lane changing behavior also considers various traffic scenarios and road configurations. For example, in~\cite{kita1999merging,hidas2005merge}, the lane changing behavior in a merging scenario is studied. Recently, a series of papers have been presented that study  microscopic lane changing behavior using game theory, such as~\cite{talebpour2015modeling, meng2016dynamic}, which have brought new perspectives and insights to the field.

\begin{figure}
    \centering
    \includegraphics[width=0.55\textwidth,height=7cm]{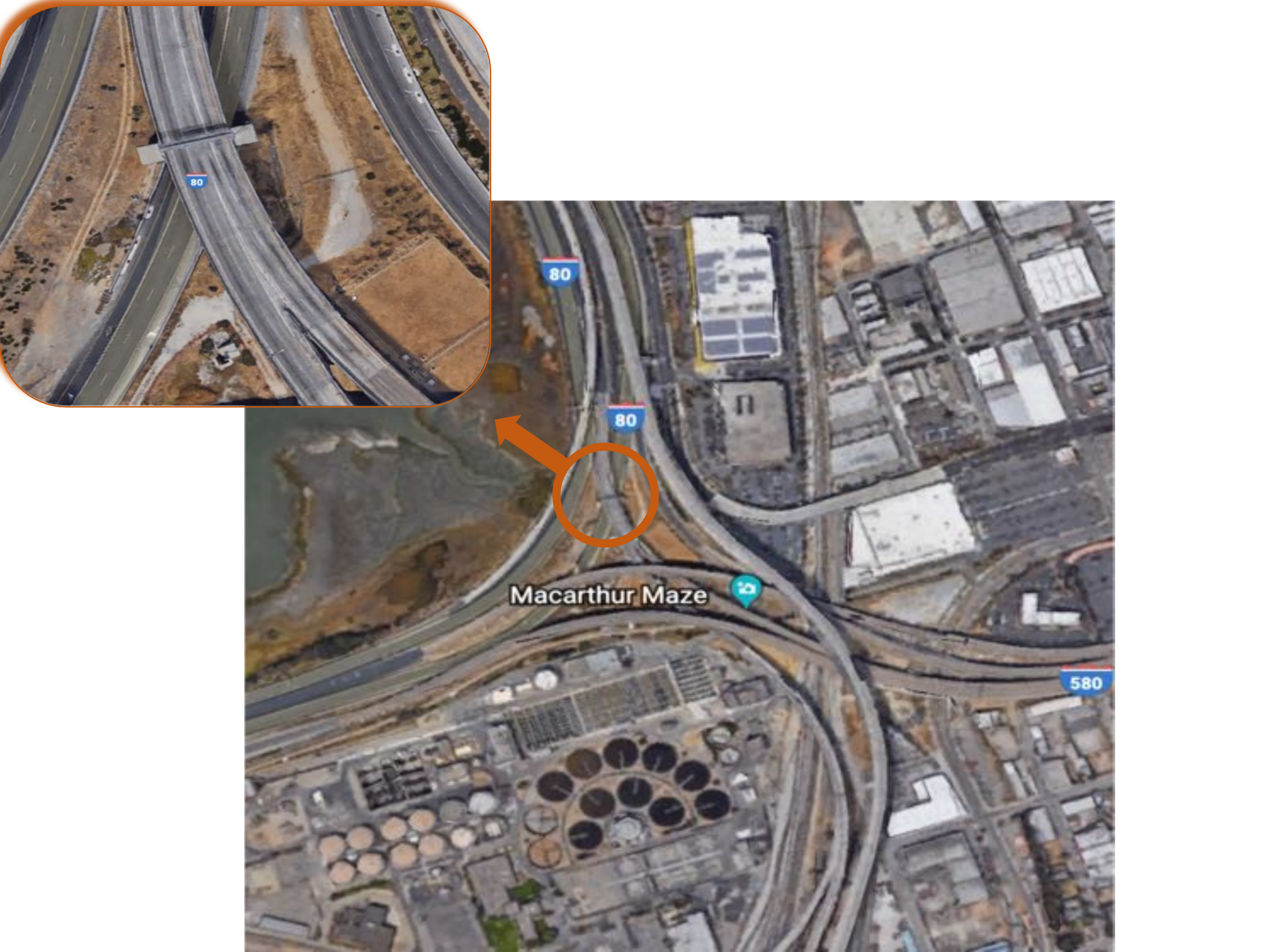}
    \caption{Satellite imagery from Google Earth of the MacArthur Maze with an enlarged view of a typical diverge with a bifurcating lane.}\label{fig:mac}
\end{figure}

Another considerable part of related work addresses the ``macroscopic'' impact of vehicles' lane changing behavior, i.e., how lane changes affect other vehicles and the aggregate traffic flow.
In~\cite{glchang1991macrocharac,brackstone1998lane}, researchers studied the macroscopic characteristics that could affect vehicles' lane changing behavior. Subsequently, in~\cite{daganzo2002merge}, lane changing behaviors of vehicles entering a freeway on-ramp and how it affects the onset of congestion are examined. In~\cite{coifman2005lane}, lane changing behaviors are also explored from a macroscopic perspective and it is shown that lane changing behaviors frequently cause the well known  ``freeway capacity drop" phenomenon. Subsequently,~\cite{laval2006lane,laval2007impacts} analyzed the  macroscopic impacts qualitatively and quantitatively of lane changing behaviors focused on the section of lanes away from freeway diverges, by modeling lane changing vehicles as particles linking interactive streams on different lanes. In our previous work~\cite{negar2018bypassing}, a scenario where vehicles bypass at the end of the diverge is studied, and the macroscopic choice behavior of vehicles of such a process is modeled as a Wardrop equilibrium~\cite{wardrop1952some}. The resulting model shows an impressive predictive power and can be easily calibrated.

In this paper, we extend the framework in~\cite{negar2018bypassing} to analyze another commonly encountered traffic diverging scenario that is a frequent cause of bottlenecks: {\em bifurcating} lanes at traffic diverges. Figure~\ref{fig:diverge} illustrates this diverge scenario, where the center lane ``b" bifurcates such that vehicles in that lane must choose either to turn left or right.
Bifurcating lanes are encountered in (poorly designed) complex distribution structures, such as the four-freeway interchange in the San Francisco Bay Area known as the MacArthur Maze (see Figure~\ref{fig:mac}), where consecutive diverges with bifurcating lanes are employed to split the  I-80 west / I-580 east (Eastshore Freeway) traffic into 1) traffic going towards  San Francisco via the I-80 west (Bay Bridge), 2) traffic going towards downtown Oakland, Walnut Creek, Hayward or Stockton via   I-580 east (MacArthur Freeway) and I-980, and 3) traffic going towards San Jose via I-880 south (Nimitz Freeway).  An aerial photograph of the MacArthur Maze is shown in Figure~\ref{fig:mac}.

Vehicles targeting one of the two exit links of the diverge face  two choices. One is to employ a feed-through lane lane (either lane "a" or lane "c" Fig.~\ref{fig:diverge}), while the other is to employ the bifurcating lane. With the reasonable assumption that drivers choose their routes in a \textit{selfish} manner in order to minimize their travel time or effort, vehicles would only employ the bifurcating middle lane in order to  save time or effort, as  compared to using the exit's respective feed--through lane, and vice versa. It turns out that vehicular lane choice behavior for this scenario has its own interesting characteristics, which to the best of our knowledge,  has not been previously addressed in the literature. In this paper, we first derive a model that describes the decision making process encountered by drivers at such traffic diverges and then we obtain the macroscopic lane choices made by drivers by solving the model's corresponding Wardrop equilibrium. It should be emphasized that the calibration of our  model only requires traffic flow information, which is realistically attainable.

This paper is organized as follows.  In Section~\ref{sec:model}, we first provide a detailed description of the notation used throughout the paper and subsequently derive a model that describes the decision making process of drivers at traffic diverges with middle bifurcating lanes. In Section~\ref{sec:eq_prop}, we establish the existence and uniqueness of the Wardrop equilibrium introduced by our model. In Section~\ref{sec:model_val}, we describe our model calibration and validation process using microscopic traffic simulation data. 
Finally, in Section~\ref{sec:future}, we conclude our work and briefly discuss  future directions.



%% file: model.tex
\section{The Model} \label{sec:model}

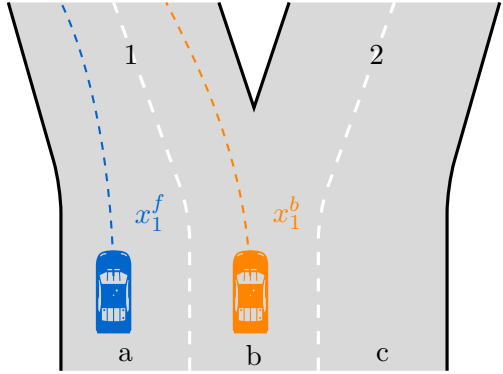
\begin{figure}
\centering
\begin{tikzpicture}[scale=0.7]
  \fill[gray!30] (0,0) [rounded corners = 10pt]-- (0,3.5) [rounded corners = 0pt]-- (-1,7)  -- (3,7) -- (11/3,5)-- (13/3,7)--(25/3,7) [rounded corners = 10pt]--(22/3,3.5)-- (22/3,0);

  \draw[very thick,rounded corners=10pt] (0,0) -- (0,3.5)-- (-1,7);
  \draw[very thick,rounded corners=10pt] (25/3,7) --(22/3,3.5)-- (22/3,0);
  \draw[very thick] (3,7) -- (11/3,5)-- (13/3,7);

  \draw[very thick, dash pattern=on 5pt off 5pt, white] (22/9,0)[rounded corners = 10pt] -- (22/9,3)-- (1,7);
  \draw[very thick, dash pattern=on 5pt off 5pt, white] (44/9,0)[rounded corners = 10pt] -- (44/9,3) -- (19/3,7);

  \node at (4/3,6) {1};  
  \node at (6,6) {2};

  \node at (11/9,0.3) {a};
  \node at (33/9,0.3) {b};
  \node at (55/9,0.3) {c};
  
  \node[scale=0.4,rotate=270] at (1,1.5) {\includegraphics{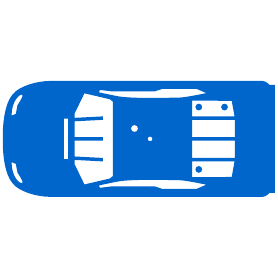}};
  \node[scale=0.4,rotate=270] at (3.6,1.5) {\includegraphics{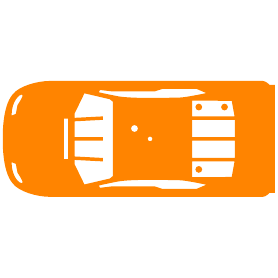}};

  \draw[thick, dashed, darkblue] (1,1.5)  .. controls (1,4) and (0.5,6) .. (0,7);
  \draw[thick, dashed, carorange] (3.6,1.5)  .. controls (3.6,4) and (2.5,6) .. (2,7);

  \node[darkblue,scale=1] at (1.7,3) {$x^f_1$};
  \node[carorange,scale=1] at (4.3,3) {$x^b_1$};

\end{tikzpicture}
\caption{Problem setting: a traffic diverge with a bifurcating lane targeting two exit links.}
\label{fig:diverge}
\end{figure}

In this paper, we consider a vehicular traffic diverge with a middle bifurcating lane. Such diverges are commonly used in a variety of modern transportation roadways, such as the well--known MacArthur Maze (see Figure~\ref{fig:mac}), in the San Francisco Bay Area. A schematic of such a traffic diverge  is shown in Figure~\ref{fig:diverge}. The diverge has two exit links, respectively denoted as link $1$ and link $2$, each consisting of two lanes, and a single entry link with three lanes respectively denoted by the letters $a,\,b,$ and $c$. Vehicles traveling along feed-through lanes  $a$ and $c$  respectively exit directly to links $1$ and  $2$. However, vehicles traveling along the middle bifurcating lane $b$ can exit to either link. Thus, drivers  upstream of such a traffic diverge need to decide between two options. They can either use the feed--through lane that is designed to exclusively  serve its corresponding destination link, or conversely they can use the middle bifurcating lane,  which targets both exit links and is shared by vehicles targeting both exit links. For example, referring to Fig.~\ref{fig:diverge}, vehicles targeting  link $1$ may choose to utilize feed--through lane $a$ or the bifurcating lane $b$. To model lane changing behavior upstream of this diverge, we consider 4 classes of vehicles. For either exit link, vehicles are either feed--through lane users or bifurcating lane users.

Here, some  key points should be clarified. First, in this paper we model  the macroscopic lane changing behavior of a given vehicular flow, instead of attempting to model the decision making process of  individual vehicles. Thus, features related to individual vehicles' choices are not considered. 
Second, we will set the flow demands of both exit links within a feasible range, in order to exclude  potential teleporting behavior of vehicles (i.e. vehicles disappear at the upstream link and then re--appear at a link downstream of the diverge), which is observed in several microscopic traffic models. 
Only bifurcating lane choice behavior will be studied in this work.
Third, we assume that the capacity of either exit link is large enough to accommodate the corresponding demand. The impact of capacity drop downstream of the diverge is not considered in this work.

Let $I = \{1,2\}$ be the index set of exit links at the diverge, and let $L = \{a,b,c\}$ be the index set of the entry link's lanes. At the diverge in Figure~\ref{fig:diverge}, lane $a$ is the feed--through lane targeting exit link $1$. Lane $c$ is the feed--through lane targeting exit link $2$. Lane $b$ is the middle bifurcating lane targeting both link $1$ and $2$. We assume that the total demand for the diverge is fixed and given. For each exit link $i\in I$, let $d_i$ be the demand of vehicles targeting exit link $i$, and let $q_i := \frac{d_i}{\sum_{i\in I}d_i}$ be the normalized demand of vehicles with destination link $i$. We collect the normalized demands and let $\mathbf{Q} := (q_i: i \in I)$ be the normalized demand configuration vector. For a diverge with the exit index set $I$ and the normalized demand configuration vector $\mathbf{Q}$, we should have $\sum_{i\in I}q_i=1$. 

For each exit link $i\in I$, let $n^f_i$ be the exact flow of feed--through lane users with destination link $i$ and let $n^b_i$ be the exact flow of bifurcating lane users with destination link $i$. For each $i \in I$, let $x^f_i:=\frac{n^f_i}{\sum_{i\in I}d_i}$ be the proportion of feed--through lane users with destination link $i$. Likewise, for each $i \in I$, we let $x^b_i:=\frac{n^b_i}{\sum_{i\in I}d_i}$ be the proportion of bifurcating lane users with destination link $i$. We then collect the proportions of the four classes of vehicles transiting through the diverge into the vector
$\mathbf{x} := (x_i^f, x_i^b: i \in I)$. For a given normalized demand configuration vector $\mathbf{Q}$, we will use our model to predict the flow distribution vector $\mathbf{x}$. A flow vector $\mathbf{x}$ is feasible if it is non--negative and it satisfies  flow conservation:
\begin{equation}\label{eq:flowcons}
    \begin{aligned}
    &q_i = x_i^f + x_i^b, \quad \forall i \in I, \\
    &x^f_i \geq 0, x^b_i \geq 0, \quad \forall i \in I.
    \end{aligned}
\end{equation}

We assume all vehicles are \emph{selfish} in that drivers will choose the route that minimizes their travel cost, which will be defined subsequently. That is to say, in this bifurcating lane scenario, vehicles would only choose the bifurcating lane when the cost experienced when traveling through the bifurcating lane is less than the cost experienced when traveling through the feed--through lane. We assume that vehicles of the same class, which are vehicles targeting the same destination link and utilizing the same entry lane, experience the same cost.  We now model the cost experienced by each class of vehicles. For each exit link $i\in I$, let $J_i^f$ denote the cost experienced by the feed--through lane users with destination link $i$, let $J_i^b$ denote the cost experienced by the bifurcating lane users with destination link $i$. We now postulate that
\begin{align}\label{eq:J_i^f_ori}
J_i^f(\mathbf{x}) &= C_i^f x_i^f \,.
\end{align}
For each exit link $i\in I$, we let $C_i^f$  represent the cost incurred by  feed--through lane users targeting exit link $i$. For each exit link $i\in I$, we assume that $C_i^f>0$ and it is a parameter related to the intrinsic features of the utilized feed--through lane, including the geometry, speed limits and  other factors. To be specific in Figure~\ref{fig:diverge}, for exit link $1$, $C_1^f$ should be a parameter related to the intrinsic features of the feed--through lane $a$. For each exit link $i\in I$, since the feed--through lane is only shared by the feed--through lane users that exit through link $i$, the cost experienced by all users should be proportional to the proportion of the feed--through lane users that exit through link $i$. Thus, for each exit link $i\in I$, we let the cost, $J_i^f$, to be the product of $C_i^f$ and $x_i^f$, as described by Eq.~(\ref{eq:J_i^f_ori}).

Let us now focus on the cost experienced by users of the middle bifurcating lane. For exit links $i\neq j\in I$, we model the cost experienced by the bifurcating lane users with destination link $i$ as
\begin{align}\label{eq:J_i^b_ori}
J_i^b(\mathbf{x}) &= C^b \left(\lambda_i x_i^b + \mu_i x_j^b \right) + \nu x_i^b x_j^b,
\end{align}
where $C^b$ is a parameter characterizing the cost incurred by  bifurcating lane users targeting either exit link. Likewise to feed--through lanes, we assume $C^b>0$ and that it is a  parameter related to the intrinsic features of the bifurcating lane. Since the bifurcating lane is shared by the bifurcating lane users for either exit link, the cost experienced by the users should be proportional to the sum of the proportions of the bifurcating lane users for either exit. 
Vehicles travelling along the bifurcating lane must take either of the two exit links at the end of the diverge, which might give rise to a sudden capacity increase for the bifurcating lane users targeting either exit link. This expected capacity increase would reduce the cost experienced by the bifurcating lane users. To account for this phenomenon, we define for each exit link $i$ two positive parameters $\lambda_i\leq 1$ and $\mu_i\leq 1$. These parameters  respectively  characterize the possible capacity increase effect on bifurcating lane users targeting exit link $i$ incurred by bifurcating lane users with the same destination link and the effect incurred by bifurcating lane users with a different destination link. If the effect is the same regardless of the inconsistency of the destination link, we should have $\lambda_i=\mu_i$; otherwise, $\lambda_i \ne \mu_i$. If either of the capacity increase effect is negligible, we will have $\lambda_i= 1$ or $\mu_i= 1$.

The second term in Eq.~(\ref{eq:J_i^b_ori}), which has the positive constant parameter $\nu$ is used to account for the detrimental effect on travel cost induced by the destination heterogeneity of the bifurcating lane users. This detrimental effect should increase when either of the two bifurcating lane vehicular flows increase. Therefore, we utilize the product between the two bifurcating lane proportions $\nu x_i^b x_j^b$, $i \ne j \in I$ in this term.

Thus, Eqs.~(\ref{eq:J_i^f_ori}) and (\ref{eq:J_i^b_ori}) for $i\ne j \in I$, are used to describe the costs experienced by vehicles traveling through the diverge. We collect these cost parameters in these equations, and define $\mathbf{C} = (C_i^f,C^b,\lambda_i,\mu_i,\nu: i \in I)$ to be the cost coefficient vector characterizing the diverge.

Having modeled the costs~\eqref{eq:J_i^f_ori} and~\eqref{eq:J_i^b_ori}, we will model vehicles' choice behavior on the macroscopic level. We assume that once a lane choice is made by the vehicle, it stays on its chosen lane. Since we have assumed that all vehicles are \emph{selfish}, vehicles would only choose the bifurcating lane when the cost experienced by traveling on the bifurcating lane is smaller than the cost experienced by traveling on the corresponding feed--through lane. Using the notations in our model, we can say, at the equilibrium of vehicles' choice behavior, if $x_i^b>0$, then we must have $J_i^b(\mathbf{x})\leq J_i^f(\mathbf{x})$; likewise, if $x_i^f>0$, then we must have $J_i^f(\mathbf{x})\leq J_i^b(\mathbf{x})$. Therefore, at an equilibrium of our model, if $J_i^f(\mathbf{x}) > J_i^b(\mathbf{x})$, then $x_i^f=0$; if $J_i^f(\mathbf{x}) < J_i^b(\mathbf{x})$, then $x_i^b=0$; only if $J_i^f(\mathbf{x}) = J_i^b(\mathbf{x})$, $x_i^b$ and $x_i^f$ may both be nonzero. These conditions can be formulated as a Wardrop equilibrium~\cite{wardrop1952some}. Now, let $\mathbf{C} = (C_i^f,C^b,\lambda_i,\mu_i,\nu: i \in I)$ be the cost coefficient vector and $\mathbf{Q} = (q_i: i \in I)$ be the normalized demand configuration vector. Let $G = (\mathbf{Q}, \bb{C})$ be a tuple configuring a traffic diverge in Figure~\ref{fig:diverge}, we interpret the above equilibrium conditions of our model and give the formal definition of the equilibrium of our model:
\begin{definition}\label{def:wdp}
For a given $G = (\mathbf{Q}, \bb{C})$, a flow distribution vector $\mathbf{x}$ is an equilibrium if and only if for every $i \neq j \in I$, we have
\begin{equation}\label{eq:eq_def}
    \begin{aligned}
    x_i^f (J_i^f(\mathbf{x}) - J_i^b(\mathbf{x})) &\leq 0 ,\\
    x_i^b (J_i^b(\mathbf{x}) - J_i^f(\mathbf{x})) &\leq 0.
    \end{aligned}
\end{equation}
\end{definition}

Now that we have modeled the cost experienced by each class of vehicles and model the resulting choice equilibrium as a Wardrop equilibrium described in Definition~\ref{def:wdp}, we can use this model to predict the proportion of bifurcating lane users and feed--through lane users for either exit link.

%% file: eq_properties.tex
\section{Equilibrium Properties}\label{sec:eq_prop}

In this section, we will first establish the existence of the equilibrium induced by our model. Then, we will derive the sufficient conditions under which the existing equilibrium is guaranteed to be unique. Therefore, once the sufficient conditions are met, our model could be applied to the prediction of the proportions of bifurcating lane users and feed--through lane users.

\subsection{Equilibrium Existence}

We will first directly give a proposition based on the existence theorem stated and proved in~\cite{braess1979existence}. 
\begin{proposition}\label{prop:wdpexist}
Given a tuple $G = (\mathbf{Q},\mathbf{C})$ configuring a traffic diverge in Figure~\ref{fig:diverge}, for each exit link $i\in I$, if each of the cost functions $J_i^f(\mathbf{x}), J_i^b(\mathbf{x})$ is continuous and monotone in $\mathbf{x}$, there exists at least one Wardrop equilibrium (as described in Definition~\ref{def:wdp}) for $G$. 
\end{proposition}

From Equations~\eqref{eq:J_i^f_ori} and~\eqref{eq:J_i^b_ori}, we observe that $J_i^b(\mathbf{x})$ and $J_i^b(\mathbf{x})$ are both continuous and monotone in the sense of non-decreasing in $\mathbf{x}$. Thus, by Proposition~\ref{prop:wdpexist}, we conclude the existence of the equilibrium described in Definition~\ref{def:wdp}.

\subsection{Equilibrium Uniqueness}
As for the uniqueness of the induced equilibrium, we will use a similar method as what we have stated in our previous work~\cite{negar2018bypassing}.
The basic idea is that we will first construct an equivalent Nash equilibrium of our equilibrium model. Then we will prove under certain conditions the uniqueness of the constructed Nash equilibrium. Therefore, the uniqueness of the Wardrop equilibrium induced by our model is concluded by equivalence.

For any tuple $G = (\mathbf{Q},\mathbf{C})$ configuring a traffic diverge in Figure~\ref{fig:diverge}, we construct a two--player auxiliary game $\tilde{G} = \langle I, A, (\tilde{J}_i: i \in I) \rangle$. Here, $I=\{1,2\}$ is the index set of our players. Let $A = A_1 \times A_2$ be the action space, $A_i = [0,q_i]$ be the action set of player $i$, and $\tilde{J}_i$ be the cost associated with each player $i \in I$. Let $\bb{y} = (y_i, i \in I)$ be the vector of actions taken by the players of the game $\tilde{G}$. To further build the correspondence, for each player $i\in I$, we let
\begin{align}\label{eq:corres}
y_i=x_i^b.
\end{align}
Then, for the cost associated with each player $i$, we define
\begin{align}\label{eq:nash_cost}
\tilde{J}_i(\mathbf{y}) := \left( J_i^f(\bb{x}) - J_i^b(\bb{x})\right)^2.
\end{align}
Next, we employ the definition of the constructed Nash equilibrium stated in~\cite{negar2018bypassing}:
\begin{definition} \label{def:nash}
For the auxiliary game $\tilde{G}$, $\bb{y} = (y_i: i\in I)$ is a pure Nash equilibrium if and only if for every $i\neq j \in I$,
\begin{equation}\label{eq:nash}
    \begin{aligned}
    y_i &= B_i(y_{j})\ \\
    &= \text{arg}\text{min}_{y_i \in [0,q_i]} \tilde{J}_i(\bb{y}),
    \end{aligned}
\end{equation}
where $B_i$ is the best response function of player $i$.
\end{definition}

Now that we have constructed a Nash equilibrium, we will use the following lemma to establish the equivalence between the Wardrop equilibrium of $G$ and the Nash equilibrium of $\tilde{G}$.

\begin{lemma}\label{lem:equivalence}
A flow distribution vector $\mathbf{x} = (x_i^f, x_i^b: i \in I)$ is a Wardrop equilibrium for $G$ if and only if $\bb{y}= (x^b_i, i \in I)$ is a pure Nash equilibrium for $\tilde{G}$.
\end{lemma}

\begin{proof}
First, using flow constraints~\eqref{eq:flowcons}, let us write $J_i^f(\bb{x})$ and $J_i^b(\bb{x})$ in terms of variables $x_i^b$. For each exit link $i\in I$, we have
\begin{align}
J_i^f(\mathbf{x}) &= C_i^f (q_i-x_i^b) \label{expandJ1},\\
J_i^b(\mathbf{x}) &= C^b \left(\lambda_i x_i^b + \mu_i x_j^b\right) + \nu x_i^b x_j^b. \label{expandJ2}
\end{align}
\noindent Then, for each exit link $i\in I$, we calculate the derivatives of the costs with respect to $x_i^b$:
\begin{align}
\frac{\partial J^f_i}{\partial x_i^b} &= -C_i^f,\label{eq:djf}\\
\frac{\partial J^b_i}{\partial x_i^b} &= C^b \lambda_i+ \nu x_j^b.\label{eq:djb}
\end{align}
Notice that in Equation~\eqref{eq:djf} and~\eqref{eq:djb}, given all coefficients are positive, for each exit link $i\in I$, $\frac{\partial J^f_i}{\partial x_i^b}$ is always negative and equals to a constant, i.e., $J_i^f(\mathbf{x})$ is a linear function of $x_i^b$ with a negative slope. For each exit link $i\neq j \in I$, due to the flow constraints~\eqref{eq:flowcons}, $x_j^b$ is always nonnegative, thus $\frac{\partial J^b_i}{\partial x_i^b}$ is always positive and increases as $x_j^b$ increases. 

\begin{figure}
    \centering
    \begin{subfigure}[b]{0.3\textwidth}
       \centering \includegraphics[width=\textwidth]{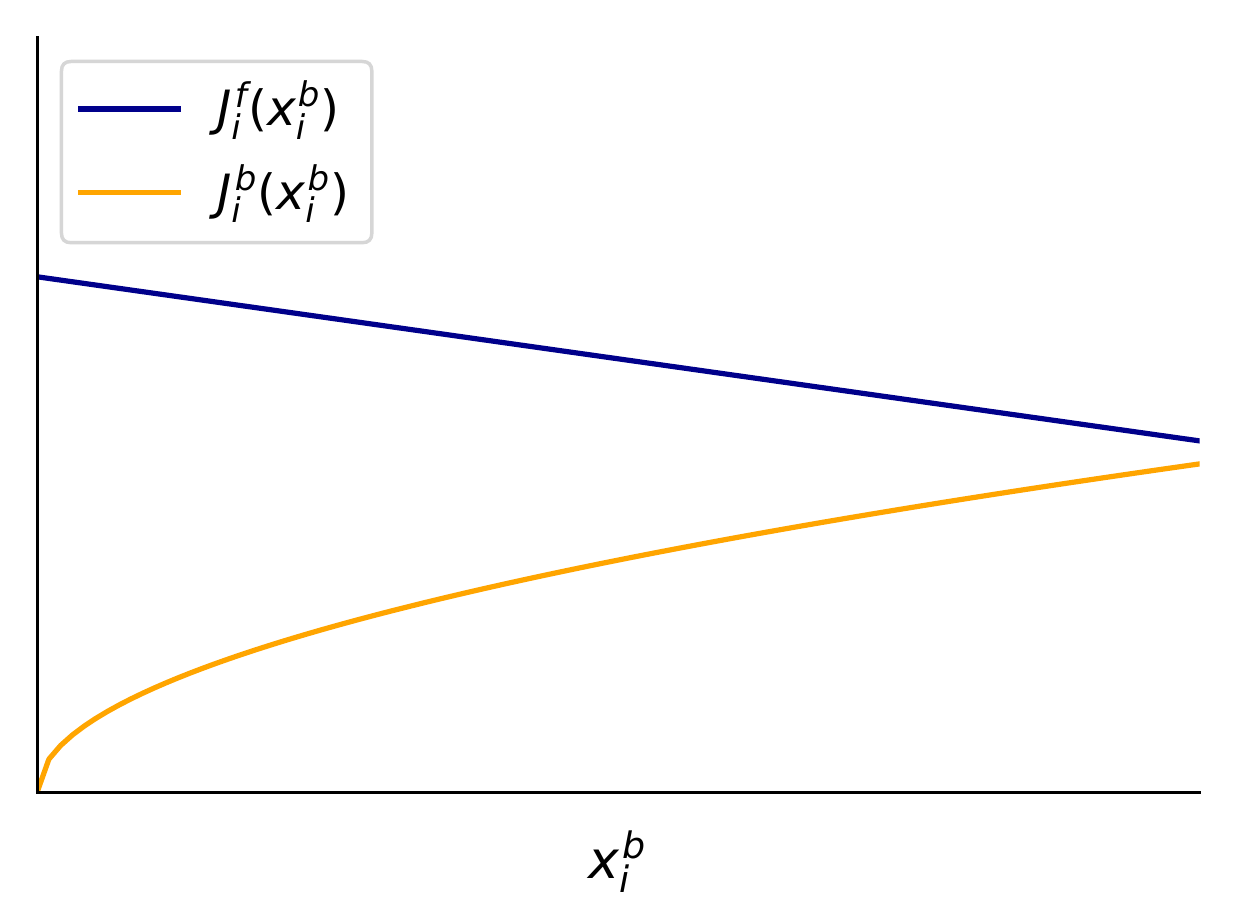}
        \caption{$J_i^f(x_i^b)>J_i^b(x_i^b)$}
    \end{subfigure}
    \begin{subfigure}[b]{0.3\textwidth}
       \centering \includegraphics[width=\textwidth]{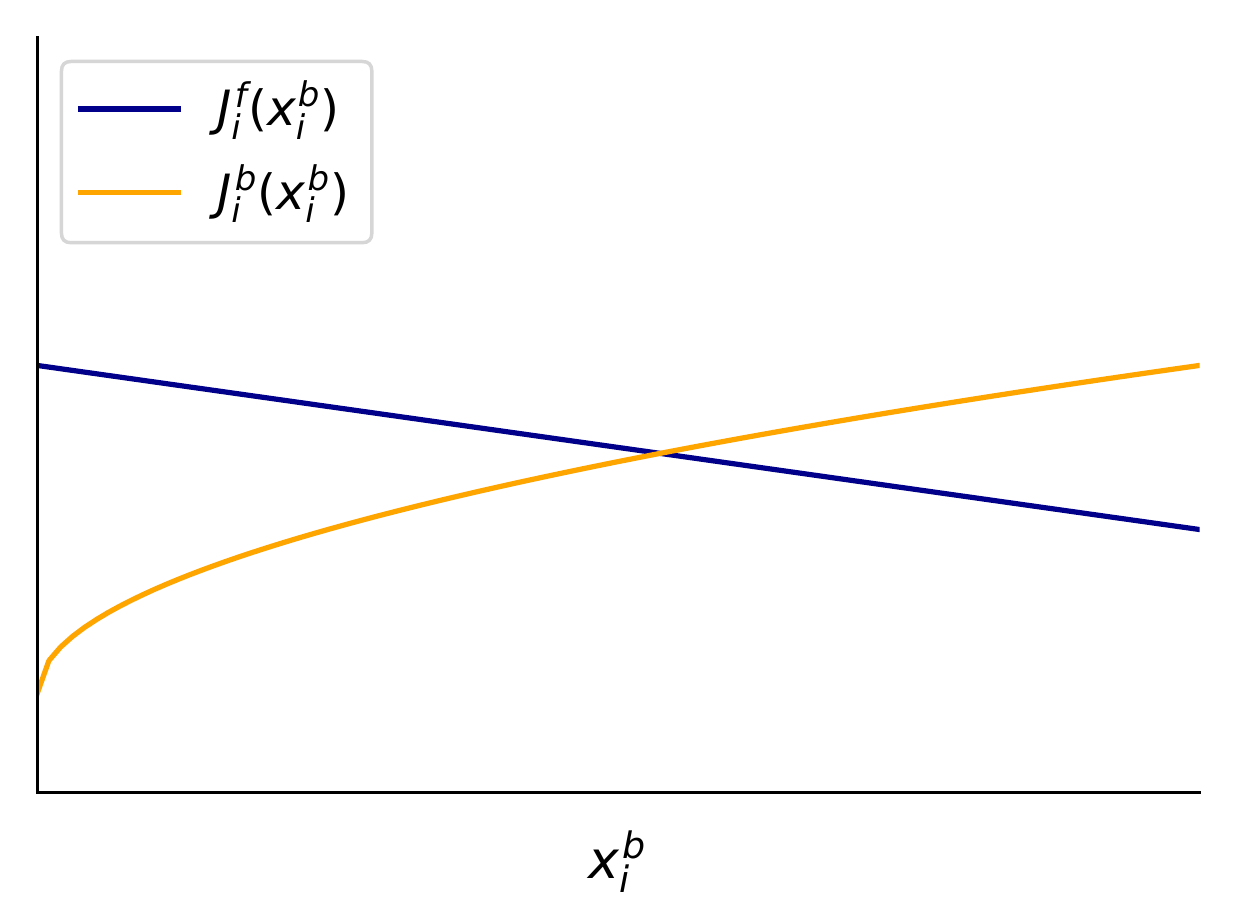}
        \caption{$J_i^f(x_i^b)$ and $J_i^b(x_i^b)$ intersect}
    \end{subfigure}
    \begin{subfigure}[b]{0.3\textwidth}
       \centering \includegraphics[width=\textwidth]{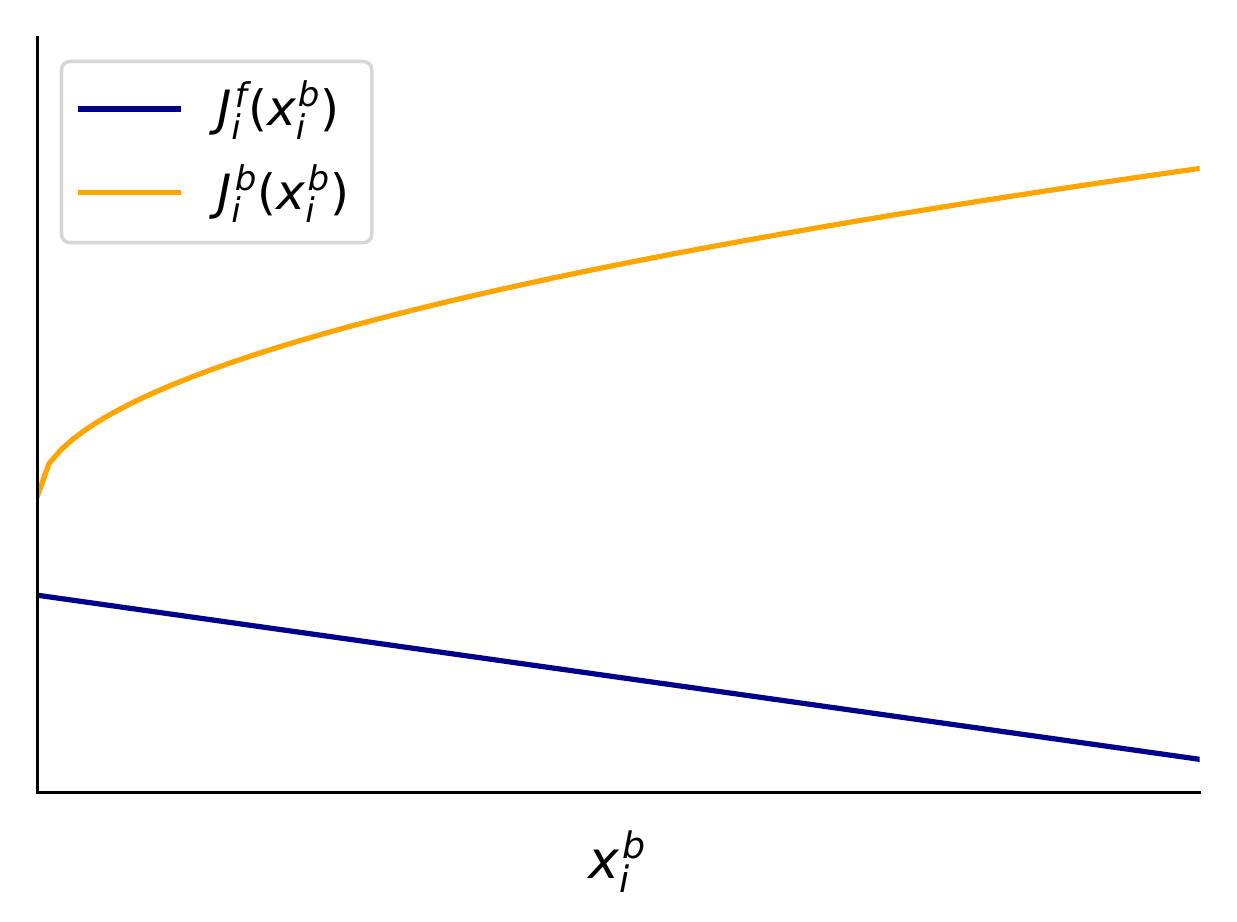}
        \caption{$J_i^f(x_i^b)<J_i^b(x_i^b)$}
    \end{subfigure}  
    \caption{Three possible sketches of $J_i^f(x_i^b)$ and $J_i^b(x_i^b)$ in the region of $x_i^b\in [0,q_i]$.}
    \label{fig:JJ}
\end{figure}

We now will show sketches of $J_i^f(\mathbf{x})$ and $J_i^b(\mathbf{x})$ in the region of $x_i^b\in [0,q_i]$. There are 3 possible cases of the sketch and we draw each possibility in Figure~\ref{fig:JJ}. We then complete the proof of equivalence case by case.

\begin{itemize}

\item Case (a): In this scenario, for every $x_i^b\in [0,q_i]$, we have $J_i^f(x_i^b)> J_i^b(x_i^b)$. To minimize the cost associated with player $i$ in Equation~\eqref{eq:nash_cost}, we have $y_i=q_i$ at the constructed Nash equilibrium. Since we let $y_i=x_i^b$, we have $x_i^b=q_i$. Then due to flow constraints~\eqref{eq:flowcons}, we have $x_i^b=q_i,x_i^f=0$, where $J_i^f(x_i^b)> J_i^b(x_i^b)$. This meets the conditions in Definition~\ref{def:wdp}, thus the constructed Nash equilibrium leads to the Wardrop equilibrium. At a Wardrop equilibrium where $J_i^f(x_i^b)> J_i^b(x_i^b)$, we can conclude that $x_i^f=0,x_i^b=q_i$. This leads to $y_i=q_i$. From plot (a), we can see that, the cost associated with player $i$ is minimized. Therefore, the Wardrop equilibrium is also a constructed Nash equilibrium. In this case, the Wardrop equilibrium is equivalent to the constructed Nash equilibrium.

\item Case (b): In this case, $J_i^f(x_i^b)$ and $J_i^b(x_i^b)$ have an intersection. Let us denote the $x_i^b$ at the intersection as ${\bar{x}^b}_i$. Then we have $J_i^f({\bar{x}^b}_i) = J_i^b({\bar{x}^b}_i)$. At the Nash equilibrium, we let $y_i=\bar{x}^b_i$, therefore the minimum possible cost associated with player $i$ is reached, which is 0. At this time, Wardrop conditions in Definition~\ref{def:wdp} are tight and met. Therefore, the constructed Nash equilibrium is also the Wardrop equilibrium. Reversely, when $J_i^f(x^b_i) = J_i^b(x^b_i)$ at the Wardrop equilibrium, from plot (b), $x^b_i$ could only equal to ${\bar{x}^b}_i$. This results in a zero cost, which is the minimum possible value of the cost. Therefore, a Nash equilibrium is reached. The equivalence of Wardrop equilibrium and its constructed Nash equilibrium is established for this case.

\item Case (c): This case is similar to case (a). For every $x_i^b\in [0,q_i]$, we have $J_i^f(x_i^b)< J_i^b(x_i^b)$. To minimize the cost associated with player $i$ in Equation~\eqref{eq:nash_cost}, we have $y_i=0$ at the constructed Nash equilibrium. Since we let $y_i=x_i^b$, we have $x_i^b=0$. Then due to flow constraints~\eqref{eq:flowcons}, we have $x_i^b=0,x_i^f=q_i$, where $J_i^f(x_i^b)< J_i^b(x_i^b)$. This meets the conditions in Definition~\ref{def:wdp}, thus the constructed Nash equilibrium leads to the Wardrop equilibrium. At a Wardrop equilibrium where $J_i^f(x_i^b)< J_i^b(x_i^b)$, we can conclude that $x_i^f=q_i,x_i^b=0$. This leads to $y_i=0$. From plot (c), we can see that the cost associated with player $i$ is minimized. Therefore, the Wardrop equilibrium is also a constructed Nash equilibrium. In this case, the Wardrop equilibrium is equivalent to the constructed Nash equilibrium.
\end{itemize}
\end{proof}
Now in all cases, we have proved the equivalence between the Wardrop equilibrium and the constructed Nash equilibrium.
Then, we use the following lemma to establish the uniqueness of the constructed Nash equilibrium.

\begin{lemma}\label{lem:uniq}
For an auxiliary game $\tilde{G}$, the Nash equilibrium flow vector $\bb{y}$ in Definition~\ref{def:nash} is unique if for each player $i\in I$:
\begin{align}
(\lambda_i - \mu_i) C^b \geq \nu-C_i^f.
\label{eq:uniqcons}
\end{align}
\end{lemma}

\begin{proof}
At a Nash equilibrium, for each player $i\neq j\in I$, we have
\begin{align}\label{B}
    y_i &= B_i(y_{j}).
\end{align}
For each player $i\neq j\in I$, we can rewrite~\eqref{B} as
\begin{align}\label{eq:eq_cond_casc}
y_i &= B_i(B_{j}(y_{i})).
\end{align}
\noindent Equation~\eqref{eq:eq_cond_casc} indicates that $\bb{y}$ is an equilibrium if and only if for every $i\neq j\in I$, $y_i$ is a fixed point for function $g(z)=B_i\left(B_{j}(z)\right)$. Thereby, the number of the fixed points of function $g(z)=B_i\left(B_{j}(z)\right)$ equals the number of equilibria. To guarantee the uniqueness of the equilibrium, we have to guarantee the uniqueness of the fixed point of function $g(z)=B_i\left(B_{j}(z)\right)$.
Notice that a fixed point $z^*$ of function $g(z)$ must satisfy
\begin{align}
    g(z^*)=z^*.
\end{align}
\noindent Therefore, the fixed point $z^*$ can be found by intersecting the identity function $h(z)=z$ and function $g(z)$. Thus, the basic idea of the following proof is to show that under~\eqref{eq:uniqcons}, for every player $i \neq j \in I$, the slope of function $g(z)=B_i\left(B_{j}(z)\right)$ is always non-negative and smaller than 1. Therefore, with $y_i\geq 0$, function $g(z)=B_i\left(B_{j}(z)\right)$ can intersect the identity function at most once, which will establish the uniqueness of the fixed point of function $g(z)=B_i\left(B_{j}(z)\right)$. Then we can conclude the uniqueness of the constructed Nash equilibrium.

First, for each player $i\neq j \in I$, we explore $\frac{d B_i(y_j)}{d y_{j}}$. Back to the three cases when we prove Lemma~\ref{lem:equivalence}, for case (a) and case (c), $y_i=B_i(y_{j})$ is always equal to $q_i$ or 0 no matter how $y_{j}$ changes, therefore, $\frac{d B_i}{d y_{j}}=0$. Then the slope of function $g(z)=B_i\left(B_{j}(z)\right)$ is always 0. Thus we only need to explore $\frac{d B_i(y_j)}{d y_{j}}$ in case (c), where $J_i^f(x_i^b)$ and $J_i^b(x_i^b)$ intersect in the region of $x_i^b\in[0,q_i]$. Remember that, we let $y_i=x_i^b$, for simplicity, we use $x_i^b$ instead of $y_i$ in the following proof. Using the same notations as the proof of Lemma~\ref{lem:equivalence}, in case (c), with a given $x_j^b$, let $J_i^f(x_i^b)$ and $J_i^b(x_i^b)$ intersect at $\bar{x}_i^b(x_j^b) \in [0,q_i]$. For $\bar{x}_i^b(x_j^b)$, we must have
\begin{align}
J_i^f(\bar{x}_i^b, {x}_j^b) - J_i^b(\bar{x}_i^b, {x}_j^b) =0.
\end{align}
Using implicit differentiation of $J_i^f(\bar{x}_i^b, {x}_j^b) - J_i^b(\bar{x}_i^b, {x}_j^b)$ with respect to $x_j^b$, we have
\begin{equation}\label{eq:impl_diff}
\begin{aligned}
& \frac{\partial}{\partial x_i^b}\left( J_i^f(\bar{x}_i^b,{x}_j^b) - J_i^b(\bar{x}_i^b,{x}_j^b) \right) \frac{ d\bar{x}_i^b(x_j^b)}{d x^b_j} + \\
&\qquad \qquad \frac{\partial}{\partial x_j^b}\left( J_i^f(\bar{x}_i^b,{x}_j^b) - J_i^b(\bar{x}_i^b,{x}_j^b) \right) = 0.
\end{aligned}
\end{equation}
\noindent Using Equations~\eqref{expandJ1} and~\eqref{expandJ2}, we have
\begin{align}
\frac{\partial J_i^f}{\partial x_j^b} &= 0,\label{eq:dj2f}\\
\frac{\partial J_i^b}{\partial x_j^b} &= C^b \mu_i + \nu x_i^b.\label{eq:dj2b}
\end{align}
Since $x_i^b\geq 0$, we can conclude that for every $i\neq j\in I$, $\frac{\partial J_i^f}{\partial x_j^b}$ is always 0 and $\frac{\partial J_i^b}{\partial x_j^b}$ is always positive. Therefore, we have
\begin{align}
\frac{\partial}{\partial x_j^b}\left( J_i^f(\bar{x}_i^b,{x}_j^b) - J_i^b(\bar{x}_i^b,{x}_j^b) \right)\leq 0.
\end{align}
From Equations~\eqref{eq:djf} and~\eqref{eq:djb}
, we have
\begin{align}
\frac{\partial}{\partial x_i^b}\left( J_i^f(\bar{x}_i^b,{x}_j^b) - J_i^b(\bar{x}_i^b,{x}_j^b) \right)\leq 0.
\end{align}
Thus, using Equation~\eqref{eq:impl_diff}, we conclude that
\begin{align}
\frac{ d\bar{x}_i^b(x_j^b)}{d x^b_j}\leq 0.
\end{align}
Intuitively, from plot (b), when $x_j^b$ increases, $J_i^f(\bar{x}_i^b,{x}_j^b)$ stays the same, whereas $J_i^b(\bar{x}_i^b,{x}_j^b)$ increases. The intersection climbs leftwards, therefore, $\bar{x}_i^b$ decreases. 

To guarantee, for every player $i \neq j \in I$, that the slope of $g(z)=B_i\left(B_{j}(z)\right)$ is always non-negative and smaller than 1, based on the chain rule, we have to guarantee that the slope of $B_{i}(y_j)$ is always non-positive and bigger than $-1$, i.e., $-1\leq \frac{ d\bar{x}_i^b(x_j^b)}{d x^b_j}\leq 0$. Now we plug Equations~\eqref{eq:djf},~\eqref{eq:djb},~\eqref{eq:dj2f} and~\eqref{eq:dj2b}
into Equation~\eqref{eq:impl_diff}, we have
\begin{equation}\label{eq:diff}
\begin{aligned}
\left(-C_i^f- C^b \lambda_i-\nu x_j^b\right) 
\frac{ d\bar{x}_i^b(x_j^b)}{d x^b_j} &+ \\
\left(-C^b \mu_i - \nu x_i^b\right) &= 0.
\end{aligned}
\end{equation}
\noindent For each player $i\neq j\in I$, to ensure $-1\leq \frac{ d\bar{x}_i^b(x_j^b)}{d x^b_j}\leq 0$, we need to guarantee
\begin{equation}
\begin{aligned}
C_i^f+ C^b \lambda_i+\nu x_j^b \geq
C^b \mu_i + \nu x_i^b.
\end{aligned}
\end{equation}
Let $M(x_i^b,x_j^b)=C_i^f+ C^b \lambda_i+\nu x_j^b -
C^b \mu_i - \nu x_i^b $, we need to guarantee that $\text{min}\ M(x_i^b,x_j^b)\geq 0$. 
Since $M(x_i^b,x_j^b)$ is negatively linear in $x_i^b$ and positively linear in $x_j^b$, the minimum possible value of $M(x_i^b,x_j^b)$ must be greater than the value of the extreme point, $M(1,0)=C_i^f+ C^b \lambda_i -C^b \mu_i -\nu $. Thus, we just need to guarantee that $M(1,0)\geq0$, which is as we stated in Lemma~\ref{lem:uniq}.

Thus, for every player $i \neq j \in I$, under condition~\eqref{eq:uniqcons}, the slope of $g(z)=B_i\left(B_{j}(z)\right)$ is always nonnegative and smaller than 1. Therefore, with $y_i\geq 0$, $g(z)=B_i\left(B_{j}(z)\right)$ can intersect the identity line at most once, which establishes the uniqueness of the fixed point of function $g(z)=B_i\left(B_{j}(z)\right)$. Thus, we can conclude the uniqueness of the constructed Nash equilibrium.
\end{proof}

We then give the following theorem to establish the uniqueness of the Wardrop equilibrium as described in Definition~\ref{def:wdp}.
\begin{theorem}\label{thm:uniq}
For a game $G = (\bb{Q}, \bb{C})$, the equilibrium flow vector $\mathbf{x}$ in Definition~\ref{def:wdp} is unique if for each exit link $i\in I$:
\begin{align}\label{eq:uni}
(\lambda_i - \mu_i) C^b \geq \nu-C_i^f.
\end{align}
\end{theorem}
\begin{proof}
From Lemma~\ref{lem:uniq}, we know that the constructed Nash equilibrium in Definition~\ref{def:nash} is unique under condition~\eqref{eq:uni}. By Lemma~\ref{lem:equivalence}, we conclude that the constructed Nash equilibrium is equivalent to the Wardrop equilibrium in the sense of Definition~\ref{def:wdp}. Thus, we can conclude that the Wardrop equilibrium as described in Definition~\ref{def:wdp} is unique.
\end{proof}

Notice that Theorem~\ref{thm:uniq} only gives a sufficient but not necessary condition of the uniqueness of the equilibrium. This implies if condition~\eqref{eq:uni} is met, we can guarantee the uniqueness of the equilibrium; however, if condition~\eqref{eq:uni} is not met, it is also possible that the Wardrop equilibrium is unique.

%% file: model_validation.tex
\section{Simulation Studies} \label{sec:model_val}
Now that we have characterized the existence and uniqueness of the equilibrium induced by our model, we are going to test the performance of our model in terms of how it accurately describes steady state vehicular flow data generated by a micro-simulation flow model. 

In this paper, we generate vehicular steady state flow data for model calibration and validation using the traffic microscopic simulation software SUMO~\cite{SUMO2012}, which is commonly utilized by the transportation community. In the simulations, to ensure the reliability of the generated data, we set the SUMO car following model to be the default Krauss model, which is also known as the stochastic version of the Gipps' model. To be specific, the Krauss model supports stochastic driving behavior by setting an imperfection parameter \textit{sigma}. The imperfection parameter \textit{sigma}, which ranges from 0 to 1, represents the degree of randomness of vehicles' behavior. When \textit{sigma} is set to nonzero, drivers will randomly vary their speed. In the simulations, we set \textit{sigma} to a default value of 0.5, in order to realistically mimic vehicle randomness. An enlarged view of the established diverge in SUMO is shown in Figure~\ref{fig:sim_en}. 
To ensure that the data we collect truly reveals the equilibrium state, we set the entry link to be sufficiently long  and only recorded the proportions of different classes of vehicles ($\mathbf{x} := (x_i^f, x_i^b: i \in I)$), downstream of the diverge when the simulation has run for a sufficiently long time to reach the steady state.

\begin{figure}
    \centering
    \includegraphics[width=0.4\textwidth,height=6cm]{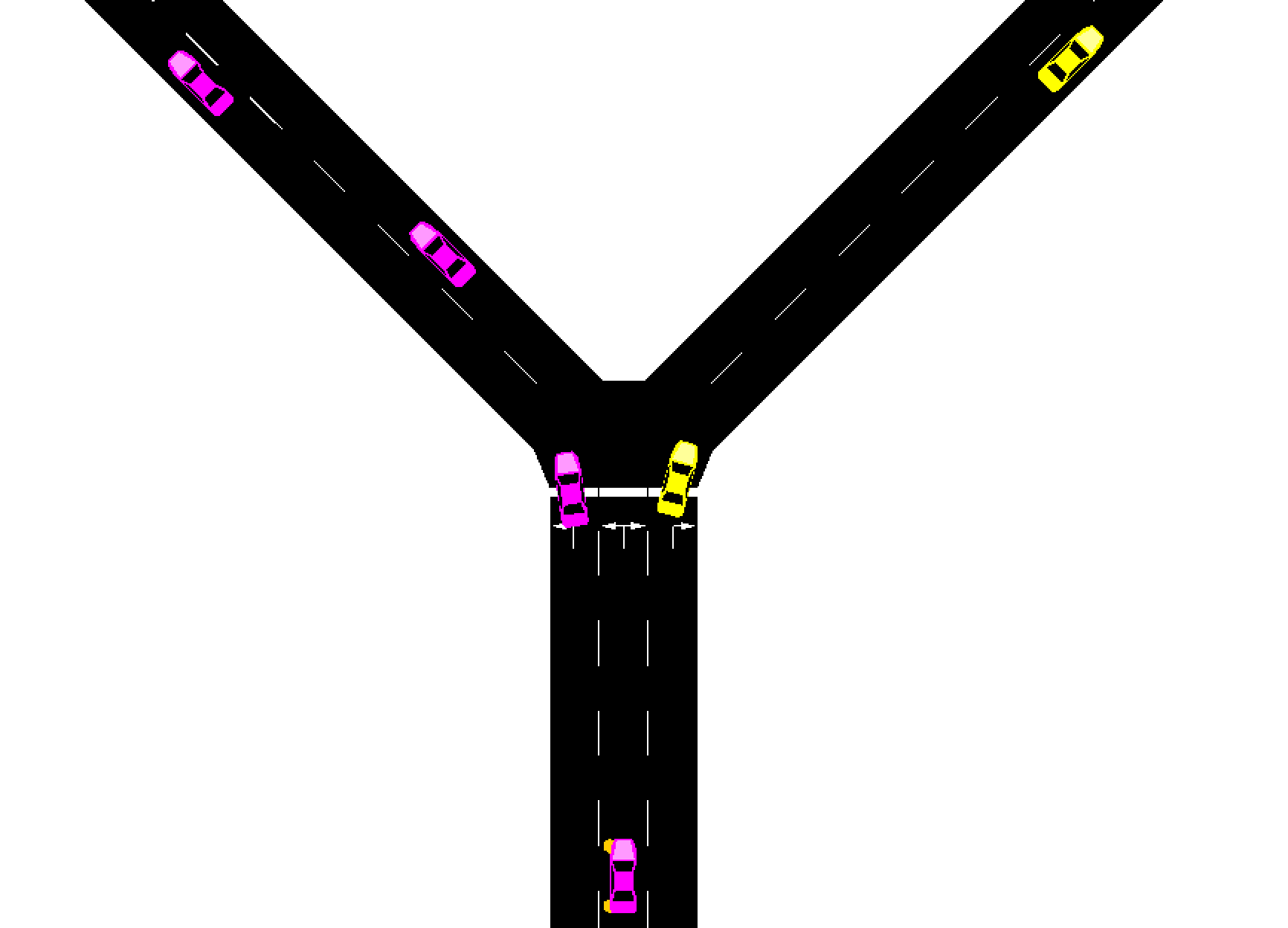}
    \caption{An enlarged view of the traffic diverge with a bifurcating lane in SUMO.}\label{fig:sim_en}
\end{figure}


\subsection{Model Calibration}

Notice that the cost coefficient vector $\mathbf{C} = (C_i^f,C^b,\lambda_i,\mu_i,\nu: i \in I)$ is related to the intrinsic features of a diverge. Thus, it is necessary to first calibrate the cost coefficient vector $\bf{C}$ in the model for a given diverge before we use the model for prediction.
For the diverge shown in Figure~\ref{fig:sim_en}, let the exit link index set be $I_{\text{exit}} = \{1,2 \}$. We define $K$ to be the total number of data points that we need for calibration. For each data point, we run the simulation once, until it reaches an equilibrium state. We pick an appropriate total demand, $D$, of vehicles entering the diverge and fix it for all $K$ simulations. For every simulation, we randomize the demand configuration for either exit link. In the $k^{\text{th}}$ simulation to generate the $k^{\text{th}}$ data point, we let $d_1^k$ be the demand of vehicles targeting exit link 1 and similarly, let $d_2^k$ be the demand of vehicles targeting exit link 2. For each $1\leq k\leq K$, we should have $d_1^k+d_2^k=D$. Now, for each $1\leq k\leq K$, we define $q_1^k:=\frac{q_1^k}{D}$ and $q_2^k:=\frac{q_2^k}{D}$ to be the normalized demand for exit link 1 and exit link 2 and we should have $q_1^k+q_2^k=1$. Then we let $\mathbf{Q}^k:=\{q_1^k, q_2^k \}$ to be the $k^{\text{th}}$ flow configuration vector for the $k^{\text{th}}$ simulation. After the simulation has reached the equilibrium, for each exit link $i\in I_{\text{exit}}$, we record the proportions of the feed--through lane users $(x_i^f)^k$ and the bifurcating lane users $(x_i^b)^k$. For each $1\leq k\leq K$, we collect the proportions in the flow distribution vector $\mathbf{x}^k=\{(x_i^f)^k,(x_i^b)^k:i\in I_{\text{exit}}\}$. Now, for each $1\leq k\leq K$, we define a tuple $(\mathbf{Q}^k,x^k)$. Then using these $K$ tuples, we employ the method developed in~\cite{negar2018bypassing} for calibration, which we briefly describe below.

We want to find a cost coefficient vector $\mathbf{C} = (C_i^f,C^b,\lambda_i,\mu_i,\nu: i \in I_{\text{exit}})$ that can enable as many as possible of the $K$ data points to meet the conditions in Definition~\ref{def:wdp}. To deal with the variational inequalities when encoding the conditions in Definition~\ref{def:wdp} for each data point, for every $1\leq k\leq K$ and $i \in I_{\text{exit}}$, we define binary variables $(e_i^f)^k$ and $(e_i^b)^k$ as
\begin{linenomath*}
\begin{subequations}\label{eq:impl_cons}
\begin{gather}
\begin{align}
(x_i^f)^k (J_i^f(\mathbf{x}^k) - J_i^b(\mathbf{x}^k)) \leq 0 &\Longleftrightarrow (e_i^f)^k = 0, \\
(x_i^f)^k (J_i^f(\mathbf{x}^k) - J_i^b(\mathbf{x}^k)) > 0 &\Longleftrightarrow (e_i^f)^k = 1, \\
(x_i^b)^k (J_i^b(\mathbf{x}^k) - J_i^f(\mathbf{x}^k)) \leq 0 &\Longleftrightarrow (e_i^b)^k = 0, \\
(x_i^b)^k (J_i^b(\mathbf{x}^k) - J_i^f(\mathbf{x}^k)) > 0 &\Longleftrightarrow (e_i^b)^k = 1.
\end{align} 
\end{gather}
\end{subequations}
\end{linenomath*}

This way, we use the binary variables $(e_i^f)^k$ and $(e_i^b)^k$ to indicate for $(x_i^f)^k$ and $(x_i^b)^k$ in each data point whether the conditions in Definition~\ref{def:wdp} are violated. To optimize for $\mathbf{C}$, we will minimize the sum of binary variables $(e_i^f)^k$ and $(e_i^b)^k$ for all $i\in I_{\text{exit}}$ and $1\leq k\leq K$. Moreover, to solve the optimization problem, we transfer the constraints in Equations~\eqref{eq:impl_cons} as in~\cite{raman2014model}. We set $T$ as a large positive number, and we set $\epsilon$ as a small positive number close to zero. Then the constraints in Equations~\eqref{eq:impl_cons} can be transferred as below:

\begin{linenomath*}
\begin{subequations}\label{eq:simp_impl}
\begin{gather}
\begin{align}
(x_i^f)^k (J_i^f(\mathbf{x}^k) - J_i^b(\mathbf{x}^k)) &\leq T (e_i^f)^k - \epsilon,  \\
-(x_i^f)^k (J_i^f(\mathbf{x}^k) - J_i^b(\mathbf{x}^k)) &\leq T (1-(e_i^f)^k) - \epsilon,  \\
(x_i^b)^k (J_i^b(\mathbf{x}^k) - J_i^f(\mathbf{x}^k)) &\leq T (e_i^b)^k - \epsilon, \\
-(x_i^b)^k (J_i^b(\mathbf{x}^k) - J_i^f(\mathbf{x}^k)) &\leq T (1-(e_i^b)^k) - \epsilon.
\end{align} 
\end{gather}
\end{subequations}
\end{linenomath*}

Then we can find the calibrated cost coefficient vector $\mathbf{C}$ by solving the mixed--integer linear program problem below:

\begin{linenomath*}\begin{equation}\label{eq:cal_opt_final}
\begin{aligned}
& \underset{\bf{C}}{\text{minimize}}
& & \sum_{1\leq k\leq K} \sum_{i \in I_{\text{exit}}} \left((e_i^f)^k + (e_i^b)^k \right) \\
& \text{subject to}
& & \text{Equations}~\eqref{eq:simp_impl},\\
&&& \bb{C}_r \geq 1.
\end{aligned}
\end{equation}\end{linenomath*}



In our simulation, the capacity per lane is 1100 \textit{vph} (vehicles per hour). We pick the total demand $D$ as 3000 \textit{vph}. We vary the demand for exit link 1 from 1150 \textit{vph} to 1850 \textit{vph}. For each demand configuration, $x_1^f$, $x_1^b$, $x_2^f$, and $x_2^b$ are recorded. In our simulation, we set the intrinsic features of every entry lane to be uniform and the geometry of two exit links to be symmetric. Thus we add the equality constraints below in our calibration process to reflect the symmetry:

\begin{equation}\label{eq:cal_sym}
    \begin{aligned}
    C_1^f &= C_2^f=C^b,\\
    \lambda_{1} &= \lambda_{2},\\
    \mu_{1} &= \mu_{2}.\\
    \end{aligned}
\end{equation}
After performing calibration process described above, we obtained the following cost coefficient vector $\mathbf{C}$:
\begin{equation}\label{eq:cal_val}
    \begin{aligned}
    C_1^f &= C_2^f=C^b=1.45,\\
    \lambda_{1} &= \lambda_{2}=0.87,\\
    \mu_{1} &= \mu_{2}=0.69,\\
    \nu &=1.\\
    \end{aligned}
\end{equation}

Note that the obtained values of $\bb{C}$ satisfy~\eqref{eq:uni}, thus, with this $\bb{C}$, we can predict the unique equilibrium $\mathbf{x}$ for each flow configuration $\mathbf{Q}$. Also, since for $i\in I_{\text{exit}}$, we have $\lambda_i< 1$ and $\mu_i< 1$, the capacity increase effect on bifurcating lane users is validated.

\subsection{Model Validation}

\begin{figure}
    \centering
    \begin{subfigure}[b]{0.45\textwidth}
       \centering \includegraphics[width=\textwidth]{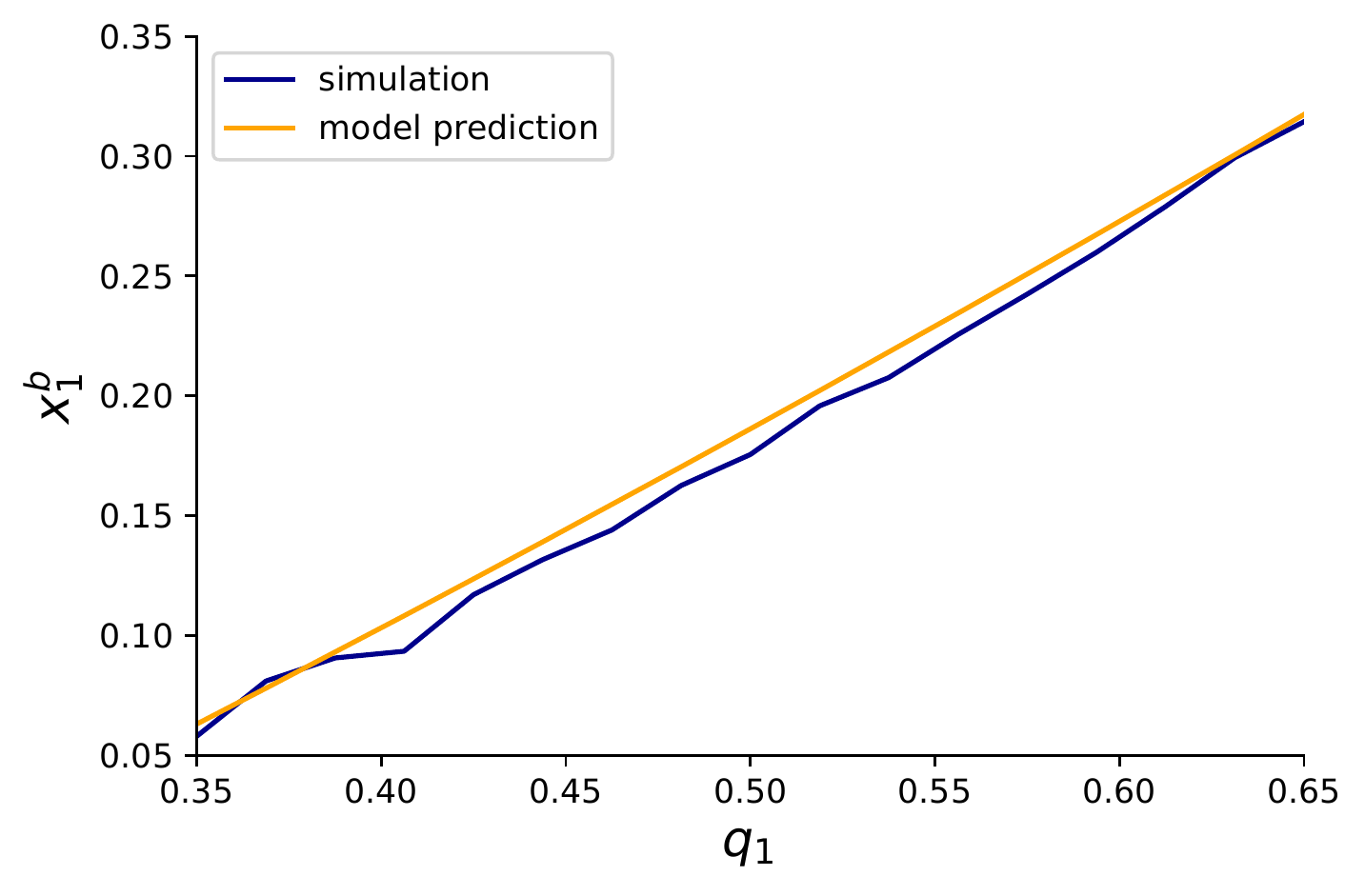}
        \caption{$x_1^b$}
    \end{subfigure}\qquad \qquad
    \begin{subfigure}[b]{0.45\textwidth}
       \centering \includegraphics[width=\textwidth]{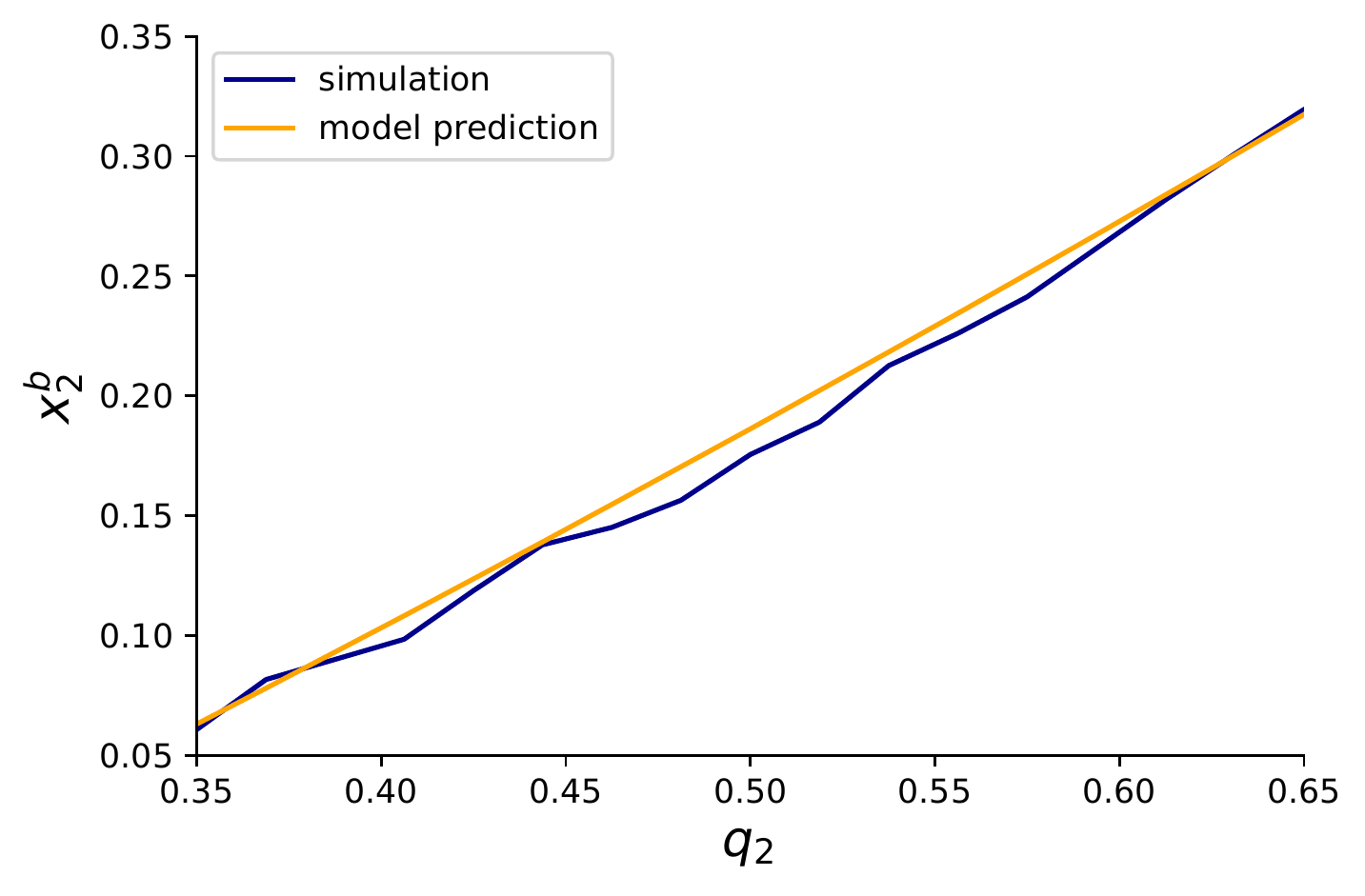}
        \caption{$x_2^b$}
    \end{subfigure}    
    \caption{Model prediction of the proportion of bifurcating lane users, $x_i^b$ is compared to simulation generated data.}\label{fig:model_pred}
\end{figure}

Having obtained the calibrated cost coefficient vector $\mathbf{C}$ in~\eqref{eq:cal_val}, we proceeded to validate our model using independently obtained simulation data from SUMO.
We validate our model under the total demand of 3200~\textit{vph}. As Figure~\ref{fig:model_pred} shows, our model successfully predicts the proportion of bifurcating lane users for either destination link. It is an obvious linear relationship which is consistent with our intuition. When the normalized demand for the same exit link increases, the proportion of bifurcating lane users increases due to the increasing cost for taking the feed--through lane designed exclusively for the exit link.
The simulation results show an impressive accuracy of our model in the prediction of vehicles' aggregate lane choice behavior. We also obtained similar results when the total demand was varied.

%% file: conclusion.tex
\section{Conclusion and Future Work} \label{sec:future}
In this work, we extended our previous work~\cite{negar2018bypassing} to another commonly encountered traffic diverge scenario. We assumed that vehicles are selfish and built a macroscopic model of vehicles' aggregating lane choice behavior at diverges with a middle bifurcating lane using Wardrop conditions. We then proved the existence and uniqueness of the resulting Wardrop equilibrium. Next, we used a microscopic traffic simulation software, SUMO, to generate data to calibrate and validate our model. The calibration process is shown to be easy. In the end, the validation results turned out to be promising, and in the future, we are looking forward to validating our model using real world data and under other similar traffic scenarios such as left-turning slots.